\documentclass[conference,letterpaper]{IEEEtran}

\IEEEoverridecommandlockouts
\usepackage[utf8]{inputenc} 
\usepackage[T1]{fontenc}
\usepackage{url}
\usepackage{ifthen}
\usepackage[cmex10]{amsmath} 
\usepackage{cite}
\usepackage{amsmath,amssymb,amsfonts}
\usepackage{graphicx}
\usepackage{textcomp}
\usepackage{xcolor}
\usepackage{booktabs}
\usepackage{array}
\usepackage{amsthm}
\usepackage[noend, ruled, linesnumbered]{algorithm2e} 
\usepackage{bm}
\usepackage{booktabs}
\usepackage{inconsolata}
\usepackage{setspace}
\usepackage[colorlinks,bookmarksopen,bookmarksnumbered,citecolor=blue, linkcolor=blue, urlcolor=blue]{hyperref}
\usepackage[caption=false,font=footnotesize]{subfig}
\usepackage{multirow}           
\usepackage{threeparttable}     
\usepackage{enumerate}
\usepackage{makecell}
\usepackage{setspace}
\usepackage{xcolor}
\usepackage{pagecolor}
\usepackage{amsmath}  

\definecolor{eye-relief}{RGB}{204,232,207}

\SetKwFor{For}{For}{do}{endfor}
\SetKwIF{If}{ElseIf}{Else}{If}{then}{else if}{Else}{endif}
\SetKwFor{While}{While}{do}{endwhile}
\SetKw{Return}{Return}

\def\BibTeX{{\rm B\kern-.05em{\sc i\kern-.025em b}\kern-.08em
		T\kern-.1667em\lower.7ex\hbox{E}\kern-.125emX}}

\newtheoremstyle{noparens} 
{2pt} 
{2pt} 
{\normalfont} 
{1em} 
{\bfseries\itshape} 
{.} 
{ } 
{\thmname{#1} \thmnumber{#2}\mdseries\thmnote{ (\hspace{-0.25pt}#3)}} 
\theoremstyle{noparens}

\newtheorem{theorem}{{\textbf{Theorem}}}
\newtheorem{lemma}[theorem]{{\textbf{Lemma}}}

\newtheorem{example}{\textbf{Example}}

\newtheorem{definition}{\textbf{Definition}}

\makeatletter

\newcommand{\Rmnum}[1]{\expandafter\@slowromancap\romannumeral #1@}
\makeatother

\newtheoremstyle{roman} 
{2pt} 
{2pt} 
{\normalfont} 
{1em} 
{\bfseries\itshape} 
{.} 
{ } 
{\thmname{#1} \Rmnum{#2}\mdseries\thmnote{ (\hspace{-0.25pt}#3)}} 
\theoremstyle{roman}

\makeatletter
\renewenvironment{proof}[1][\proofname]{\par
	\pushQED{\qed}%
	\normalfont \topsep2\p@\@plus2\p@\relax
	\trivlist
	\item[\hskip\labelsep
	\hskip\dimexpr1.5em\relax 
	\itshape
	#1\@addpunct{:}]\ignorespaces
}{%
	\popQED\endtrivlist\@endpefalse
}
\makeatother

\begin{document}


\title{GII-Polar Codes for Block Fading Channels}

\author{
	\IEEEauthorblockN{
		Fangbo Yi\IEEEauthorrefmark{2}\IEEEauthorrefmark{3},
		Zhongjun Yang\IEEEauthorrefmark{2},
		Xinzheng He\IEEEauthorrefmark{2}\IEEEauthorrefmark{3},
		Li Chen\IEEEauthorrefmark{2}\IEEEauthorrefmark{3},
		and
		Huazi Zhang\IEEEauthorrefmark{4}\\
		\IEEEauthorblockA{\IEEEauthorrefmark{2}
			School of Electronics and Information Technology, Sun Yat-sen University, Guangzhou 510006, P. R. China\\
		}
		\IEEEauthorblockA{\IEEEauthorrefmark{3}
			Guangdong Province Key Laboratory of Information Security Technology, Guangzhou 510006, P. R. China\\
		}
		\IEEEauthorblockA{\IEEEauthorrefmark{4}
			Hangzhou Research Center, Huawei Technologies Co., Ltd., Hangzhou 310052, P. R. China\\
		}
		\IEEEauthorblockA{Email: yifb@mail2.sysu.edu.cn, yangzhj59@mail2.sysu.edu.cn,  hexzh9@mail2.sysu.edu.cn, \\chenli55@mail.sysu.edu.cn, zhanghuazi@huawei.com
		}
}}

\maketitle

\begin{abstract}

	Polar codes are proven to be capacity-achieving codes, being gradually practiced in wireless communications.
	However, their successive cancellation (SC) and successive cancellation list (SCL) decoding incur latency challenge especially for long codes.
	This paper proposes the generalized integrated interleaved (GII)-polar codes for block Rayleigh fading channels, yielding both reduced decoding latency and competent decoding performance.
	Under the GII paradigm, two consecutive polar codewords of length $N$ are virtually coupled through a nested codebook which has a lower dimension, capable of correcting more errors.
	The component polar codes are known as interleaves of a GII-polar code.
	If decoding of an interleave fails, it can be projected into the decoding of the nested one,
	enabling richer error patterns to be corrected by each interleave.
	Since decoding of the interleaves can be performed in parallel,
	GII-polar codes yield a reduced decoding latency over a single polar code of length $2N$.
	Our simulation results validate both the decoding latency and performance merits of the proposed coding scheme under block Rayleigh fading channel.

\end{abstract}


\section{Introduction}

\IEEEPARstart{P}{olar}
codes, invented by Ar{\i}kan \cite{Arikan2009-polar_codes}, are the first coding scheme mathematically proven to achieve the capacity of symmetric binary-input discrete memoryless channels (B-DMCs) by the use of successive cancellation (SC) decoding.
It was adopted by the 5G communication standard \cite{5G_Standard}.
However, SC decoding suffers from poor performance caused by partially polarized channels especially under short-to-moderate codeword length regimes.
To achieve a better performance, the SC list (SCL) decoding \cite{SCL-Vardy2015, SCL-Niu2012_SCL} and the cyclic redundancy check (CRC) aided SCL (CA-SCL) \cite{SCL-Niu2012-CA_SCL} decoding were proposed.
The above mentioned algorithms can significantly enhance the SC decoding performance.
However, it requires path expansion and sorting operations during the estimation of each information bit.
Therefore, it inevitably leads to the latency challenge, especially under larger codeword length regimes.
Other effective methods to improve the polar decoding performance contain the bit-flipping \cite{SCF-Afisiadis2014_orig, SCF-Chandesris2018-Dynamic_SCF, SCF-Ercan2019} and the soft information perturbation \cite{Yang2025-HPSC, Yang2026, SCL-Wang2023-Perturbation}, but at the cost of increased decoding complexity.

Recently, several researches have focused on the application of polar codes over the Rayleigh fading channel (RFC).
In \cite{Bravo-Santos2013}, by recursively computing the Bhattacharyya parameter,  a reliability evaluation of polarized subchannels was proposed.
In \cite{Trifonov2015-Rayleigh}, a low complexity construction method was introduced.
It functions by approximating and tracking the error probabilities of the polarized subchannels during the polarization.
Based on the independent and identically distributed over the RFC, an efficient design rule was proposed for both the polar codes and polar lattices \cite{Liu2016-Polar_Fading}.
Subsequently, a systematic polar‑spectrum‑based framework was proposed in \cite{Niu2020-Polar_Spectrum_Fading}, which is used for the design of polar codes over the fast RFC.
This framework was further extended to the block RFC (BRFC) \cite{Niu2021-Polar_Spectrum_Block_Fading}.
However, the decoding latency challenge remains for transmission over RFC.
Meanwhile, fading diversity has not been exploited through the possible redesign of codes.

Generalized integrated interleaved (GII) codes \cite{GII-Wu2017-org} are the promising codes that can yield both a low decoding latency and yet a good decoding performance.  
It consists of $v$ interleaves (also known as component codes) of length $N$.
Linear combinations of the $v$ interleaves can generate $z$ nested codewords with stronger error-correction capabilities, where $z\leq v$.
Decoding of all interleaves can be performed in parallel.
If there remain interleaves undecoded, the nested codes can be incurred to further recover them. 
Hence, GII codes can achieve both a high decoding throughput and an advanced decoding performance.
For the GII codes, hard-decision decoding was proposed in \cite{GII-Wu2017-org, GII-Zhang2021-Fast_Nested_Key_Equation_Solvers, GII-Zhang2020-Key_Equation_Solvers, GII-Zhang2019-hardware}, along with its performance analysis in \cite{GII-Wang2021-Performance_Analysis}.
In \cite{GII-Zhang2023-Soft_Decision, GII-He2025-Performance_Analysis, GII-He2025-Performance_Analysis_and_Enhanced_Chase}, soft-decision decoding of GII codes was proposed.
However, so far, only cyclic codes, including BCH codes and Reed-Solomon (RS) codes, have been considered as interleaves for constructing GII codes.
It remains to be studied how to construct a GII code using modern codes as interleaves and how they will perform.

In this paper, a novel GII-polar code is proposed.
It yields both reduced decoding latency and competent decoding performance.
Under the GII paradigm, two consecutive polar codewords of length $N$, which are also known as interleaves, are virtually coupled through a nested codebook.
The constructed nested code has a lower dimension, capable of correcting more errors.
If decoding of an interleave fails, it can be projected into the decoding of the nested code,
enabling richer error patterns to be corrected by each interleave.
Since decoding of all interleaves can be performed in parallel, a GII-polar code yields a reduced decoding latency over a single polar code of length $2N$.
Numerical results show that, over the BRFC, the proposed GII-polar codes outperform the conventional polar codes, producing both the decoding latency  and performance merits.

\section{Preliminaries}

This section provides the preliminaries of the paper, including polar codes, its SCL decoding and the GII codes.
In this paper, we use $\mathcal{B}[i]$ to denote the $i$-th element of an ordered set $\mathcal{B}$, where $0 \le i < |\mathcal{B}|$.
Let $a_{i,j}$ denote the vector $( {{a_i}, {a_{i+1}}, \ldots ,{a_{j}}} )$, and $ \oplus $ denote the bitwise XOR operation.
Further let $\mathbb{F}_q = \{0, \alpha^0, \alpha^1, \alpha^{q-2}\}$ denote the finite field of size $q$, where $\alpha$ is a primitive element.

\subsection{Polar Codes} \label{Sec-2-polar_codes}

Polar codes with length $N=2^n$ and dimension $K$, denoted as $\mathcal{P}(N,K)$, are linear block codes with a rate $R=K/N$.
After rate profiling, the $K$ most reliable subchannels, whose the index set is denoted as $\mathcal{A}$ and called the information set, are used to transmit information bits.
Let $\mathcal{A}^{c}$ denote the index set of the remaining $N-K$ least reliable subchannels.
Generally, it is known as the frozen set.
For polar coding, subchannels indexed by $\mathcal{A}^{c}$ are utilized to transmit the frozen bit.
It is often determined as zero and known by the polar decoder.
Note that $\mathcal{A} \cap {\mathcal{A}^c} = \emptyset$.

For polar encoding, the $K$-bit information vector, denoted as $m_{0,{K - 1}}=( {m_0}, {m_{1}},\ldots, $ ${m_{K-1}} ) \in {\mathbb{F}_2^K}$ with the $N-K$ frozen bits form an input vector $u_{0,{N - 1}}$.
Afterwards, the $N$-bit codeword $c_{0,{N - 1}}$ can be determined by
\begin{equation}	\label{eq2-encode-1}
	c_{0,{N - 1}} = u_{0,{N - 1}}{{\bf{G}}_{N}},
\end{equation}
where ${{\bf{G}}_{N}} = {\bf{F}}^{ \otimes n}$, $\bf{F}={[\begin{smallmatrix} 1 & 0 \\ 1 & 1 \end{smallmatrix}]}$ denotes the Ar{\i}kan kernel, and ${ \otimes n}$ is the $n$-th Kronecker product.
Let ${\bf{g}}_i$ denote row-$i$ of $\mathbf{G}_{{N}}$, where $i \in \mathcal{A} \cup {\mathcal{A}^c}$.
According to \cite{Arikan2009-polar_codes}, eq. (\ref{eq2-encode-1}) can be redefined as
\begin{equation}	\label{eq2-encode-2}
	\begin{aligned}
		c_{0,{N - 1}} &= \sum_{i \in \mathcal{A}} u_i \mathbf{g}_i \\
		&= m_{0,{K - 1}} \mathbf{G}_{\mathcal{A}},
	\end{aligned}
\end{equation}
where ${{\mathbf{G}}_\mathcal{A}} = {\left[ {{{\mathbf{g}}_{\mathcal{A}[0]}},{{\mathbf{g}}_{\mathcal{A}[1]}}, \cdots ,{{\mathbf{g}}_{\mathcal{A}[K - 1]}}} \right]^{\text{T}}}$.

In this paper, it is assumed that codeword is transmitted through the BRFC using binary phase shift keying (BPSK).
Let $y_i \in \mathbb{R}$ denote the received symbol corresponds to $c_i$.
It can be computed as
\begin{equation}
	y_i = h_i (1-2c_i) + w_i,
\end{equation}
where $ {h_i} \in \left[ {0, + \infty } \right) $ is the fading coefficient, and $w_i \sim \mathcal{N}(0, \sigma^2)$ denotes the channel noise with its power $\sigma^2$.
Let further define ${\cal L}\left( {{y_i}} \right)$ as the received LLR of $y_i$.
It can be determined as
 \begin{equation}	\label{eq-cho}
 	\mathcal{L}\left(y_{i}\right) = \ln \frac{{\Pr \left( {{y_i} \mid {c_i} = 0} \right)}}{{\Pr \left( {{y_i} \mid {c_i} = 1} \right)}},
 \end{equation}
where $\Pr \left( {{y_i} \mid {c_i}} \right)$ denotes the channel observation of $c_i$.
In the BRFC, it can be simplified into $\mathcal{L}\left(y_{i}\right) = ({2y_i {|h_i|}^2}) / {\sigma^2}$.


\subsection{SC and SCL Decoding} \label{Sec-2-SC}

The SC decoding can be viewed as a traversal process over the binary tree that represents the codebook.
The tree contains $n+1$ layers and $N$ leaves.
Let $\hat{u}_i$ denote the estimation of message bit $u_i$.
Note that, for frozen bits, they are fixed as zeros, i.e., $\hat{u}_i = 0$ for any $i \in \mathcal{A}^c$.
During the decoding process, the estimation of information bit $\hat{u}_i$ is made upon by its SC decoding \textit{a posteriori} LLR, denoted as $\mathcal{L}\left(u_{i}\right)$.
That says
\begin{equation} \label{eq2-decision}
	{{\hat u}_i} = \left\{ {\begin{array}{*{20}{l}}
			{1,} &  { \text{if} \ {\cal L}\left( {{u_i}} \right) \le 0 \ \text{and} \ i \in \mathcal{A}}; \\
			{0,} & {\text{otherwise}.}
	\end{array}} \right.
\end{equation}

Let us consider a polar code with $N=2$.
It has the simplest tree structure.
When obtaining the received symbol vector $y_0^1$ and further determining $\mathcal{L}\left(y_{0}\right)$ and $\mathcal{L}\left(y_{1}\right)$ as in (\ref{eq-cho}), the SC decoder performs the $f$-function to compute $\mathcal{L}\left(u_{0}\right)$ as
\begin{equation} \label{eq2-f}
	{\cal L}\left( {{u_0}} \right) = {\rm{sign}}\left( {{\cal L}\left( {{y_0}} \right) {\cal L}\left( {{y_1}} \right)} \right) \cdot \min \left( {\left| {{\cal L}\left( {{y_0}} \right)} \right|,\left| {{\cal L}\left( {{y_1}} \right)} \right|} \right).
\end{equation}
Afterwards, based on (\ref{eq2-decision}), $\hat{u}_0$ can be determined.
By utilizing $\hat{u}_0$, the decoder further invokes the $g$-function to estimate $u_1$.
It can be expressed as 
\begin{equation} \label{eq2-g}
	{\cal L}\left( {{u_1}} \right) = \left( {1 - 2{{\hat u}_0}} \right) {\cal L}\left( {{y_0}} \right) + {\cal L}\left( {{y_1}} \right).
\end{equation}

For polar codes with larger codeword lengths, e.g., $N \ge 2$, the SC decoding recursively performs the $f$ and $g$ functions for determining ${\cal L}\left( {{u_i}} \right)$.
E.g., when $N=32$, computing $\mathcal{L}\left(u_{1}\right)$ invokes the $f$-function four times with an additional $g$-function.
The SC decoding terminates once all information bits have been estimated.
In this case, all leaves are visited.
Moreover, when determining $\hat{u}_0^{N-1}$, the estimation of codeword, i.e., $\hat{c}_0^{N-1}$ can be further obtained as in (\ref{eq2-encode-1}).

The SCL decoding can be viewed as performing the $L$ independent SC decoders in parallel. %
It preserves multiple candidate paths, i.e., the estimations of input vector, within the decoding process in order to increase the possibility of successfully decoding.
When estimating the information bit, each candidate path is split according to the two possible bit estimations, yielding $2L$ decoding candidates.
E.g., when estimating $\hat u_i$, where $i \in \mathcal{A}$, the current decoding path $\hat u_0^{i-1}$ is split into the candidate paths $(\hat u_0^{i-1}, 0)$ and $(\hat u_0^{i-1}, 1)$.
Among these candidate paths, only the $L$ most reliable paths are retained.
For the remaining $L$ paths, they are pruned.
Once all information bits are decoded, the final estimation $\hat{u}_0^{N-1}$ is selected from the surviving paths with the largest reliability.

\subsection{GII Structure} \label{Sec-2-GII_codes}

In general, GII coding structure consists of $v$ interleaves and $z$ nested interleaves, each of length $N$ and $v \ge z$.
Interleave-$j$ is written as $c_{0,{N-1}}^{(j)}=( {c_{0}^{(j)}, {c_{1}^{(j)}}, \ldots ,{c_{N-1}^{(j)}}} )$, where $j \in \left\{ {0,1, \ldots ,v - 1} \right\}$.
Let $\mathbb{C}_{i}$ denote the codebook with dimension $K_i$, where $i \in \left\{ {0,1, \ldots ,z} \right\}$.
Note that interleaves exhibit the nested codebooks as ${\mathbb{C}_{z}} \subseteq {\mathbb{C}_{z - 1}} \subseteq \cdots \subseteq {\mathbb{C}_1} \subseteq {\mathbb{C}_0}$, and hence ${K_0} \ge {K_1} \ge  \cdots \ge {K_{z}}$.
Let
\begin{equation}\label{eq2-3-1}
	{\mathbf{B}} =
	\begin{bmatrix}
		b_{0,0} & b_{0,1} & \cdots & b_{0,v-1} \\
		b_{1,0} & b_{1,1} & \cdots & b_{1,v-1} \\
		\vdots & \vdots & \ddots & \vdots \\
		b_{z-1,0} & b_{z-1,1} & \cdots & b_{z-1,v-1}
	\end{bmatrix}
\end{equation}
denote a coefficient matrix and ${\mathbf{B}} \in \mathbb{F}_q^{z \times v}$.
Let $\tilde{c}_{{0},{N - 1}}^{(t)} \in \mathbb{C}_{z-t}$ denote the $t$-th nested interleave, where $t \in \left\{ {0,1, \ldots ,z - 1} \right\}$.
The GII code can be defined as\cite{GII-Wu2017-org}
\begin{equation}\label{eq2-3-2}
	\begin{aligned}
		&{\left[ \tilde c_{0,N - 1}^{(0)},\, \tilde c_{0,N - 1}^{(1)},\, \cdots,\, \tilde c_{0,N - 1}^{(z - 1)} \right]^{\text{T}}} \\
		& = \mathbf{B}
		{\left[ c_{0,N - 1}^{(0)},\, c_{0,N - 1}^{(1)},\, \cdots,\, c_{0,N - 1}^{(v - 1)} \right]^{\text{T}}}
	\end{aligned}
\end{equation}
where $c_{0,N - 1}^{(0)},c_{0,N - 1}^{(1)}, \cdots ,c_{0,N - 1}^{(v - 1)} \in {\mathbb{C}_0}$ are the $v$ interleave codewords.
Hence, the generation of a nested interleave can be expressed as
\begin{equation}\label{eq2-3-3}
	\tilde c_{0,N - 1}^{(t)} = \sum\limits_{j = 0}^{v - 1} {{b_{t,j}} c_{0,N - 1}^{(j)}},
\end{equation}
where $t \in \left\{ {0,1, \ldots ,z - 1} \right\}$.
The GII code has a length of $vN$ and a dimension of $K = (v - z)K_0 + K_1 + K_2 + \cdots + K_z$.
The following Theorem \ref{theorem-1} characterizes the decoding of GII codes.

\begin{theorem}[\cite{GII-He2026}] \label{theorem-1}
	In order to realize of the GII decoding, the first $\beta$ rows and any $\beta$ columns of coefficient matrix ${\mathbf{B}}$ can form a nonsingular submatrix, where $1 \le \beta \le z $.
\end{theorem}

It is worth mentioning, if the submatrix of ${\mathbf{B}}$ is singular, there exists at least one erroneously estimated interleave that can not be corrected through the nested structure.
That says it can not be recovered by $ \tilde c_{0,N - 1}^{(0)},\tilde c_{0,N - 1}^{(1)}, \cdots$, and $\tilde c_{0,N - 1}^{(z - 1)}$.

\section{Construction of GII-Polar Codes}\label{Sec-3}

This section introduces the construction of the proposed GII-polar codes and shows their equivalence to polar codes.


\subsection{Encoding of GII-Polar Codes}

The following Lemma \ref{lemma-1} founds our proposed GII-polar code construction, explaining how can a GII code be constructed by two polar interleaves.

\begin{lemma}\label{lemma-1}

	Consider a binary GII code that is constructed by $v$ interleaves and inherited with $z$ nested interleaves.
	When $z \ge 2$, it holds that $v = 2$.

\end{lemma}

\begin{proof}

	Based on Theorem \ref{theorem-1}, when $\beta=1$, dimension of the submatrices is $1 \times 1$, ${\mathbf{B}}$ dissolves into a finite field element ${\mathbf{B}} \in \mathbb{F}_q$.
	To ensure the nonsingularity mentioned in Theorem \ref{theorem-1}, all elements in the first row are required to be nonzero.

	When $\beta=2$, dimension of the submatrices is $2 \times 2$.
	Given a set of vectors
	\begin{equation}	\label{eq3-1-lemma-proof-1}
		\{ {{{\boldsymbol{\alpha }}_0} = {{\left[ {1,0} \right]}^{\text{T}}},{{\boldsymbol{\alpha }}_1} = {{\left[ {1,{\alpha ^0}} \right]}^{\text{T}}}, \cdots ,{{\boldsymbol{\alpha }}_{q - 1}} = {{\left[ {1,{\alpha ^{q - 2}}} \right]}^{\text{T}}}} \},
	\end{equation}
	any nonsingular $2 \times 2$ matrix with nonzero elements in its first row will have column vectors that can be expressed as $\gamma {{\boldsymbol{\alpha }}_i}$, where $\gamma  \in {\mathbb{F}_q}$ and $i=0,1,\ldots,q-1$.
	Therefore, ${\mathbf{B}}$ has at most $q$ pairwise linearly independent columns.
	That says $v \le q$.

	Note that, for binary codes, $q=2$.
	In this case, given $z \ge 2$, it follows that $\beta \in \{1, 2\}$ and $v \le 2$.
	Since $v \ge z$, it holds that $z = v = 2$.
\end{proof}

Lemma \ref{lemma-1} implies that, when $z \ge 2$, a binary GII code can comprise at most two interleaves.
Therefore, despite the GII structure can accommodate an arbitrary number of $v$ interleaves, a binary GII-polar code can accommodate at most two interleaves. The proposed GII-polar codes follow this construction limit.

Let $\mathcal{P}(N,K_0)$ and $\mathcal{P}(N,K_1)$ denote two interleaves, respectively, where $K_0 > K_1$.
Again, let ${\mathbb{C}_0} $ and ${\mathbb{C}_1} $ denote their codebook, respectively, and ${\mathbb{C}_1} \subseteq {\mathbb{C}_0}$.
Let $\mathcal{A}_j$ denote the information set of the $\mathcal{P}(N,K_j)$ code, where $j \in \left\{ {0,1} \right\} $.
For a GII-polar code, it holds that ${\mathcal{A}_1} \subseteq {\mathcal{A}_0}$.
For the $\mathcal{P}(N,K_j)$ code, let $m_{{0},{{K_j} - 1}} ^{(j)} $ and $c_{{0},{N - 1}} ^{(j)}$ denote an information vector and its codeword generated.
As mentioned in Sec. \ref{Sec-2-GII_codes}, for the GII codes, all interleaves are included in $\mathbb{C}_0$.
That says
\begin{subequations}\label{eq3-1-1}
	\begin{align}
		c_{{0},{N - 1}} ^{(0)} &= m_{{0},{{K_0} - 1}} ^{(0)} \mathbf{G}_{\mathcal{A}_0}, \label{eq3-1-1a} \\
		c_{{0},{N - 1}} ^{(1)} &= \left( m_{{0},{{K_1} - 1}} ^{(1)}, r_{{0},{{K_0}-K_1 - 1}} \right) \mathbf{G}_{\mathcal{A}_0}. \label{eq3-1-1b}
	\end{align}
\end{subequations}
where $r_{{0},{{K_0}-K_1 - 1}} \in {\mathbb{F}_2^{{K_0} - {K_1}}}$ denotes the redundant vector for encoding $c_{{0},{N - 1}} ^{(1)}$.
It is designed for maintaining $c_{{0},{N - 1}} ^{(1)} \in \mathbb{C}_0$.
Based on (\ref{eq2-3-2}), the nested interleaves of a GII-polar code can be defined as
\begin{subequations}	\label{eq3-1-2}
	\begin{align}
		\tilde{c}_{0,{N - 1}}^{(0)} &= c_{{0},{N - 1}}^{(0)} \oplus c_{{0},{N - 1}}^{(1)}, \label{eq3-1-2a} \\
		\tilde{c}_{0,{N - 1}}^{(1)} &= c_{{0},{N - 1}}^{(1)},\label{eq3-1-2b}
	\end{align}
\end{subequations}
where $\tilde{c}_{0,{N - 1}}^{(0)}, \tilde{c}_{0,{N - 1}}^{(1)} \in \mathbb{C}_1$.
Hence, if either the decoding of $c_{{0},{N - 1}} ^{(0)}$ and $c_{{0},{N - 1}} ^{(1)}$ fails, decoding of $\tilde{c}_{0,{N - 1}}^{(0)}$ can be incurred. Since $\tilde{c}_{0,{N - 1}}^{(0)}\in \mathbb{C}_1$, it has a stronger error-correction capability than the interleaves.
Note that $\tilde{c}_{0,{N - 1}}^{(1)}$ is the copy of ${c}_{0,{N - 1}}^{(1)}$ without any additional error-correction capability.
For simplicity, in this work, we use ${c}_{0,{N - 1}}^{(1)}$ to denote $\tilde{c}_{0,{N - 1}}^{(1)}$ and let $\tilde{c}_{0,{N - 1}}=\tilde{c}_{0,{N - 1}}^{(0)}$.
Based on (\ref{eq2-3-2}), (\ref{eq3-1-1}), and (\ref{eq3-1-2}), the proposed GII-polar codes can be defined as the follows.

\begin{figure}[t]
	\centering
	\includegraphics[scale=0.5]{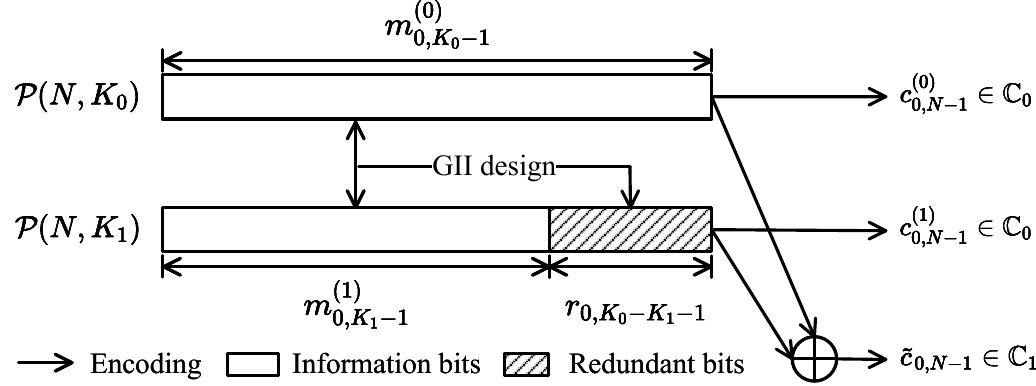}
	\caption{
		The encoding block diagram of a ${\mathcal G}\left( {N,{K_0},{K_1}} \right)$ GII-polar code.
	}
	\label{Fig.3-1}
\end{figure}

\begin{definition}\label{def-1}

	Consider a GII-polar code with length $2N$ and dimension $K_0 + K_1$, denoted as ${\mathcal G}\left( {N,{K_0},{K_1}} \right)$, its rate is $ R = \left( {{K_0} + {K_1}} \right)/2N$.
	Its codebook can be defined as
	\begin{equation}	\label{eq3-def-1}
		\begin{aligned}
			{\mathbb{C}_{{\text{GP}}}} \triangleq \{ & {c_{0,2N - 1}} = (c_{0,N - 1}^{(0)},c_{0,N - 1}^{(1)}) \mid \\
				& \ c_{0,N - 1}^{(j)} \in {\mathbb{C}_0},{\text{ }}j = \{ 0,1\} ;{{\tilde c}_{0,N - 1}} \in {\mathbb{C}_1} \},
		\end{aligned}
	\end{equation}
	where $\tilde{c}_{0,{N - 1}}$ is generated as in (\ref{eq3-1-2a}).

\end{definition}

The following Theorem \ref{theorem-2} further defines the redundant vector $r_{{0},{{K_0}-K_1 - 1}}$.

\begin{theorem}\label{theorem-2}

	Given a GII-polar code ${\mathcal G}\left( {N,{K_0},{K_1}} \right)$ with information vectors $ m_{{0},{{K_0} - 1}}^{(0)}$ and $ m_{{0},{{K_1} - 1}}^{(1)}$, respectively, there exists $r_{{0},{{K_0}-K_1 - 1}} = m_{{{K_1}},{{K_0} - 1}}^{(0)}$.

\end{theorem}

\begin{proof}
	Let $\mathcal{S}={\mathcal{A}_0}{\backslash }{\mathcal{A}_1}$, and $|\mathcal{S}|=K_0-K_1$.
	Generator matrix $\mathbf{G}_{\mathcal{A}_0}$ can be redefined as
	\begin{equation}\label{proof1-1}
		{{\mathbf{G}}_{{\mathcal{A}_0}}} = \left[ {\begin{array}{*{20}{c}}
				{{{\mathbf{G}}_{{\mathcal{A}_1}}}} \\
				{{{\mathbf{G}}_\mathcal{S}}}
		\end{array}} \right],
	\end{equation}
	where ${{\mathbf{G}}_\mathcal{S}} = {\left[ {{{\mathbf{g}}_{\mathcal{S}[0]}},{{\mathbf{g}}_{\mathcal{S}[1]}}, \cdots ,{{\mathbf{g}}_{\mathcal{S}[K_0 - K_1 - 1]}}} \right]^{\text{T}}}$.

	Based on (\ref{eq3-1-1}) - (\ref{proof1-1}), since $m_{0,{K_0-1}}^{(0)} = ( m_{0,{K_1-1}}^{(0)}, m_{K_1,{K_0-1}}^{(0)} )$,
	\begin{equation}\label{proof1-2}
		\begin{aligned}
			\tilde{c}_{0,{N - 1}}
			= & m_{0,{K_0-1}}^{(0)} \mathbf{G}_{\mathcal{A}_0}
			\oplus \left( m_{0,{K_1-1}}^{(1)}, r_{0,{K_0-K_1-1}} \right) \mathbf{G}_{\mathcal{A}_0} \\
			= & \left( m_{0,{K_1-1}}^{(0)} \oplus m_{0,{K_1-1}}^{(1)} \right) \mathbf{G}_{\mathcal{A}_1} \\
			  & \ \oplus \left( m_{K_1,{K_0-1}}^{(0)} \oplus r_{0,{K_0-K_1-1}} \right) \mathbf{G}_{\mathcal{S}}.
		\end{aligned}
	\end{equation}
	Note that $( m_{0,{K_1-1}}^{(0)} \oplus m_{0,{K_1-1}}^{(1)} ){{\mathbf{G}}_{{\mathcal{A}_1}}} \in {\mathbb{C}_1}$.
	Let $\mathbf{0}_l$ denote the all-zero vector with length-$l$.
	Given $\tilde{c}_{0,{N - 1}} \in {\mathbb{C}_1}$, $\mathcal{A}_1^c$ as the frozen set of $\mathcal{P}(N,K_1)$ and $u_{\mathcal{A}_1^c}={\mathbf{0}_{{N} - {K_1}}}$, the above equation implies that $( m_{K_1,{K_0-1}}^{(0)} \oplus r_{0,{K_0-K_1-1}} ){{\mathbf{G}}_\mathcal{S}} \triangleq {{\mathbf{0}}_N}$, i.e., $u_{\mathcal{S}} = {\mathbf{0}_{{K_0} - {K_1}}}$.
	That says $ m_{K_1,{K_0-1}}^{(0)} \oplus r_{0,{K_0-K_1-1}}  \triangleq {\mathbf{0}_{{K_0} - {K_1}}}$ and $r_{0,{K_0-K_1-1}} = m_{K_1,{K_0-1}}^{(0)}$.
\end{proof}

It can be seen that for the $\mathcal{P}(N,K_0)$ code, the redundant bits of $r_{0,{K_0-K_1-1}}$ are transmitted through the least reliable subchannels in ${\mathcal{A}_0}$, indexed by the set $\mathcal{S}$.
Fig. \ref{Fig.3-1} shows the encoding block diagram of a $\mathcal{G}(N, K_0, K_1)$ GII-polar code. The following Example \ref{eg-1} further illustrates the encoding process.

\begin{example}\label{eg-1}

	Consider a ${\mathcal G}\left( {8,{6},4} \right)$ GII-polar code with information sets ${\mathcal{A}_0} = \left\{ {7,6,5,3,4,2} \right\}$ and ${\mathcal{A}_1} = \left\{ {7,6,5,3} \right\}$.
	The information vectors are $m_{0,5}^{(0)}$ and $m_{0,3}^{(1)}$, respectively.
	Based on Theorem \ref{theorem-2}, the redundancy vector satisfies $r_{0,1} = ( m_{4}^{(0)}, m_{5}^{(0)} )$.
	After rate profiling, the input vectors are $u_{0,7}^{(0)} = (0,0, m_{5}^{(0)}, m_{3}^{(0)}, m_{4}^{(0)}, m_{2}^{(0)}, m_{1}^{(0)}, m_{0}^{(0)} )$ and $u_{0,7}^{(1)} = ( 0,0, m_{5}^{(0)}, m_{3}^{(1)}, m_{4}^{(0)}, m_{2}^{(1)}, m_{1}^{(1)}, m_{0}^{(1)} )$, respectively.
	Since $ {u_{i}^{(0)}} \oplus {u_{i}^{(1)}} = 0$ holds for $\forall i \in {\mathcal{A}_1^c}$,
	we have $(u_{0,7}^{(0)} \oplus u_{0,7}^{(1)}) {{\mathbf{G}}_8} \in {\mathbb{C}_1}$, i.e., ${c_{0,7}^{(0)} \oplus c_{0,7}^{(1)} \in {\mathbb{C}_1}}$.

\end{example}

\subsection{Equivalence to a Polar Code of Length $2N$}

Based on (\ref{eq2-3-1}) and (\ref{eq3-1-lemma-proof-1}), for the GII codes defined over $\mathbb{F}_2$, it follows that ${\bf{B}} = {[\begin{smallmatrix} 1 & 1 \\ 0 & 1 \end{smallmatrix}]}$.
Note that ${\bf{B}}^{\text{T}} = {\bf{F}}$. 
Since ${c_{0,N - 1}^{(0)}}$ and ${c_{0,N - 1}^{(1)}}$ are polar codewords with length $N$, a GII-polar code of length $2N$ can also be equivalently seen as a polar code of the same length.
Given a ${\mathcal G}\left( {N,{K_0},{K_1}} \right)$ GII-polar code, it can be interpreted as a $\mathcal{P}(2N, K_0 + K_1)$ polar code.
Based on (\ref{eq2-3-2}), codeword of the $\mathcal{P}(2N, K_0 + K_1)$ code can be expressed as $(c_{{0},{N - 1}}^{(1)} \oplus \tilde{c}_{0,{N - 1}}, c_{{0},{N - 1}}^{(0)})$, 
with information set defined as
\begin{equation}\label{eq3-1-4}
	\mathcal{A} = {\mathcal{A}_1} \cup \left\{ {{\mathcal{A}_0}\left[ 0 \right] + N,{\mathcal{A}_0}\left[ 1 \right] + N, \ldots ,{\mathcal{A}_0}\left[ {{K_0} - 1} \right] + N} \right\},
\end{equation}
where ${\mathcal{A}_0}$ and ${\mathcal{A}_1}$ are the information sets of the two interleaves within the ${\mathcal G}\left( {N,{K_0},{K_1}} \right)$ GII-polar code.
It is noteworthy that, compare with the conventional polar codes, the proposed GII-polar codes enable a more flexible decoding behavior.
The two interleaves can be decoded in parallel.
If the second decoding stage is needed, the successfully decoded interleave can provide the additional \textit{a priori} information for decoding the nested codeword.
This yields a reduced decoding latency over a single polar code of the same length.
The next section shows decoding of the GII-polar codes.

\section{Decoding of GII-Polar Codes} \label{Sec3-2-GII_dec}

For a ${\mathcal G}\left( {N,{K_0},{K_1}} \right)$ GII-polar code, its decoding can be divided into two stages.
Let $y_{0,{N-1}}^{(j)}$ denote the received symbol vector of the transmitted codeword $c_{0,{N-1}}^{(j)} \in \mathbb{C}_0$, where $j \in \left\{ {0,1} \right\}$.
In this paper, CRC codes are employed to identify the correctly estimated interleaves.
Let $CRC(\cdot)$ denote the CRC function, where an output of $true$ signifies a successful decoding.
Otherwise, if $CRC(\cdot)$ yields $false$, it implies a decoding failure.

In the first decoding stage, all interleaves can be decoded in parallel.
At this stage, the polar decoder is performed based on $y_{0,{N-1}}^{(j)}$ and $\mathcal{A}_0$.
The estimation $\hat{u}_{0,{N-1}}^{(j)}$ will be validated by a CRC code.
If $CRC(\hat{u}_{0,{N-1}}^{(j)})$ yields $true$ for $\forall j \in \left\{ {0,1} \right\}$, the decoding will terminate.
Note that if all interleaves are erroneously estimated, i.e., $CRC(\hat{u}_{0,{N-1}}^{(j)})$ yields $false$ for $\forall j \in \left\{ {0,1} \right\}$, the decoding will also terminate but fail.
Otherwise, the nested structure needs to be applied in decoding the nested interleave $\tilde{c}_{0,{N - 1}}$. 
Based on (\ref{eq3-1-2}), if it succeeds, the underlying interleave (in the first stage) can be further recovered.

The second decoding stage is invoked when there is only one interleave being successfully decoded.
That says $CRC(\hat{u}_{0,{N-1}}^{(j)})$ yields $true$ and $CRC(\hat{u}_{0,{N-1}}^{(j^*)})$ yields $false$, where ${j^*} = \left\{ {0,1} \right\} \backslash {j}$.
In this case, it can be assumed that $\hat{u}_{0,{N-1}}^{(j)} = u_{0,{N-1}}^{(j)}$ and $\hat{c}_{0,{N-1}}^{(j)} = c_{0,{N-1}}^{(j)}$.
Let $\tilde{y}_{0,{N - 1}}$ denote the virtual received symbols of the nested interleave $\tilde{c}_{0,{N - 1}}$.
With BPSK, the entry of ${{\tilde y}_i}$ can be computed as
\begin{equation}	\label{eq3-2-1}
	{{\tilde y}_i} = (1-2\hat{c}_i^{(j)}) y_i^{({j^*})}.
\end{equation}
Once obtaining $\tilde{y}_{0,{N - 1}}$, the nested interleave $\tilde{c}_{0,{N - 1}}$ can be further decoded based on $\mathcal{A}_1$.
Meanwhile, the nested input vector $\tilde{u}_{0,{N - 1}}$, i.e, $\hat{\tilde{u}}_{0,{N - 1}}$ can be obtained.
Based on (\ref{eq2-encode-1}) and (\ref{eq3-1-2a}), estimations of interleave ${c}_{0,{N-1}}^{(j^*)}$ and its input vector ${u}_{0,{N-1}}^{(j^*)}$, i.e., $\hat{c}_{0,{N-1}}^{(j^*)}$ and $\hat{u}_{0,{N-1}}^{(j^*)}$, can be obtained as
\begin{equation}\label{eq-theorem-u}
	\hat{u}_{0,{N-1}}^{(j^{*})} = \hat{\tilde{u}}_{0,{N-1}} \oplus \hat{u}_{0,{N-1}}^{(j)}.
\end{equation}
After decoding the nested interleave ${\tilde{c}}_{0,{N-1}}$, a CRC validation is performed.
If $CRC(\hat{u}_{0,{N-1}}^{(j^{*})})$ yields $true$, the current \(\hat{u}_{0,{N-1}}^{(0)}\) and \(\hat{u}_{0,{N-1}}^{(1)}\), i.e., the estimations of $m_{0,{K_0-1}}^{(0)}$ and $m_{0,{K_1-1}}^{(1)}$ will be output.
Otherwise, the decoding terminates with a failure.
In particular, it should be mentioned that the polar decoder can be realized by any existing decoding algorithm, e.g. the SCL decoding.

\section{Decoding Latency and Complexity}

In this section, both the decoding latency and complexity of GII-polar codes are analyzed.

\subsection{Decoding Latency}

The decoding latency is measured as the average number of required clock cycles (CCs) in decoding one GII-polar codeword.
Let ${{ \tau }_{{\text{PD}}}}(K)$ denote the decoding latency of the underlying polar decoder in decoding the $\mathcal{P}(N,K)$ code.
For the GII-polar decoding, as mentioned in Sec. \ref{Sec3-2-GII_dec}, two CCs are required for performing the CRC validation and one CC is required for computing the $\tilde{y}_{0,{N - 1}}$ of (\ref{eq3-2-1}).
Moreover, for the process of estimating $\hat{u}_{0,{N-1}}^{(j^{*})}$, as in (\ref{eq-theorem-u}), an extra CC is needed.
Let $T_\text{avg}$ denote the average number of stages for decoding GII-polar code, where $T_\text{avg} \in [1, 2] $.
Since all interleaves can be decoded in parallel during the first decoding stage \cite{GII-Wu2017-org}, the average latency of the GII-polar decoding can be characterized as
\begin{equation}\label{eq4-latency}
	{{ \tau }_{{\text{GPD}}}} \triangleq {{ \tau }_{{\text{PD}}}}(K_0) + ({T_\text{avg}-1}) {{ \tau }_{{\text{PD}}}}(K_1) + 4.
\end{equation}
It is noteworthy that
${T_\text{avg}}$ of $1$ and $2$ is corresponding to the best-case and worst-case decoding, respectively.

\subsection{Decoding Complexity}

Complexity of the proposed decoding attributes to the complexity of the underlying polar decoder in decoding the $\mathcal{P}(N,K)$ code, denoted as ${\Phi _{{\text{PD}}}}$ and that of computing $\tilde{y}_{0,{N - 1}}$.
Therefore, the average complexity of GII-polar decoding can be defined as
\begin{equation}\label{eq4-comp}
	{\Phi _{{\text{GPD}}}} = \left( {{T_{{\text{avg}}}} + 1} \right){\Phi _{{\text{PD}}}} + N,
\end{equation}
where $N$ denotes the complexity of computing the virtual $\tilde{y}_{0,{N - 1}}$, as in (\ref{eq3-2-1}). 
Similar to (\ref{eq4-latency}), the worst-case complexity for the GII-polar decoding can be characterized as $ 3{\Phi _{{\text{PD}}}} + N$.

\section{Simulation Results}

This section presents simulation results of the proposed GII-polar coding scheme over the BRFC.
Polar codes with $N \in \left\{ {1024,2048} \right\}$ and $R \in \left\{ {0.211, 0.25, 0.375, 0.5, 0.539} \right\}$ are considered.
The information sets are obtained through Gaussian approximation (GA) \cite{GA-Peter2012} at the \textit{design}-SNR of 2.5 dB.
The length-24 CRC code of 5G new radio \cite{5G_Standard} is employed for the decoding, whose generator
polynomial is $X^{24} + X^{23} + X^6 + X^5 + X + 1$.
In this paper, SCL decoding is considered for the GII-polar codes.
Based on \cite{SCL-Niu2012-CA_SCL} and \cite{SCL-Stimming2015-hardware_LLR_based}, ${\Phi _{{\text{PD}}}} = \mathcal{O}(LN \log_2 N)$  and ${{\tau }_{{\text{PD}}}}(K) = 2N-2 +K$.
Note that for a conventional polar code, the CA-SCL \cite{SCL-Niu2012-CA_SCL} decoding is also provided for comparison with the list size of $L \in \{2, 8\}$.
For the BRFC, it is set that
\begin{equation}
	h_i =
	\begin{cases}
		h_0, & 0 \le i < 1024;\\
		h_{1024}, & 1024 \le i < 2048,
	\end{cases}
\end{equation}
where ${h_0},{h_{1024}} \in \left[ {0, + \infty } \right)$.

\subsection{Construction and Decoding Performance}

The GII-polar code is heuristically constructed as follows.
Let $P_\text{e}(K)$ denote the SCL decoding performance of a $\mathcal{P}(N,K)$ polar code,  simulated by \cite{SCL-Niu2012-CA_SCL}.
Given a targeted rate $R$, all possible combinations of $K_0$ and $K_1$ that satisfy $(K_0 + K_1)/2N = R$ are enumerated.
Based on \cite{GII-He2025-Performance_Analysis_and_Enhanced_Chase}, SCL decoding performance of a ${\mathcal G}\left( {N, K_0, K_1} \right)$ code is
\begin{equation}\label{eq6-1}
	P_{\text{e}}^{{\text{GPD}}} = {P_{\text{e}}}{\left( {{K_0}} \right)^2} + 2{P_{\text{e}}}\left( {{K_1}} \right)\left( {1 - {P_{\text{e}}}\left( {{K_0}} \right)} \right).
\end{equation}
Therefore, the adequate choices of $K_0$ and $K_1$ can be heuristically determined.
The $\mathcal{G}(N, K_0, K_1)$ code yields the best frame error rate (FER) performance.

\begin{figure}[t]
	\centering
	\includegraphics[scale=0.32]{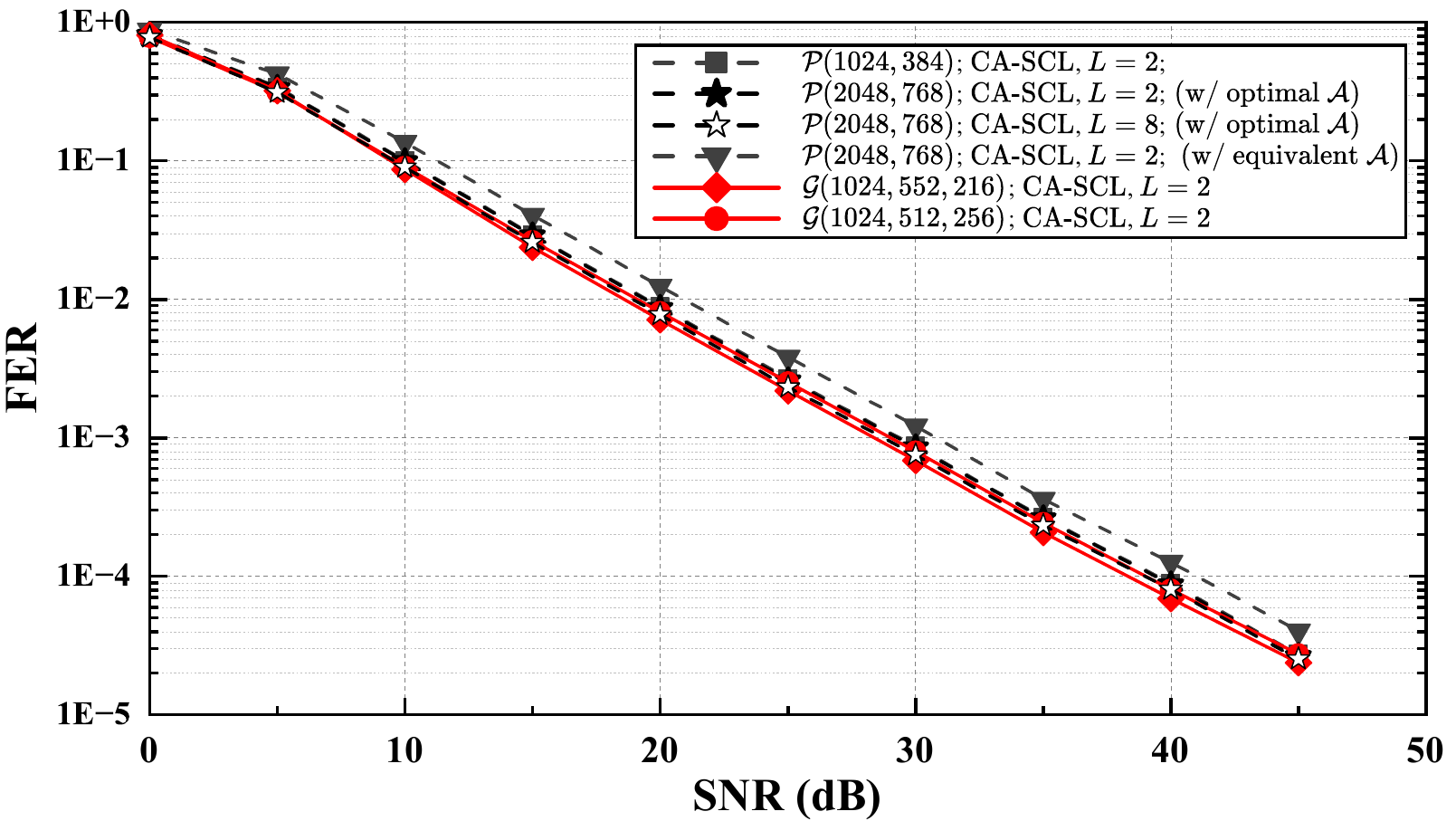}
	\caption{
		Decoding performance of the GII-polar codes.
	}
	\label{Fig.6-1}
\end{figure}

Fig. \ref{Fig.6-1} shows FER performance of the ${\mathcal G}\left( {1024, 552, 216} \right)$ and the ${\mathcal G}\left( {1024, 512, 256} \right)$ GII-polar codes.
Its SCL decoding performance is compared with those of the ${\mathcal P}\left(2048, 768\right)$ and ${\mathcal P}\left(1024, 384\right)$ codes.
All codes have a similar rate of 0.375.
For the ${\mathcal P}\left( {2048, 768} \right)$ code, constructions are conducted either through GA profiling in \cite{GA-Peter2012}, or the equivalently rate profiling, i.e., $\mathcal{A}$ is obtained as in (\ref{eq3-1-4}).
It can be seen that the GII-polar codes, e.g., ${\mathcal G}\left( {1024, 552, 216} \right)$, with the proposed parameter design, as in (\ref{eq6-1}), can achieve a better performance than the codes constructed with random parameter configurations, e.g., ${\mathcal G}\left( {1024, 512, 256} \right)$.
Therefore, the following analyses will focus on the ${\mathcal G}\left( {1024, 552, 216} \right)$ code.
It can be seen that the proposed GII-polar code outperforms the ${\mathcal P}\left( {2048, 768} \right)$ code under both rate profilings with the same SCL decoding list size.
The ${\mathcal G}\left( {1024, 552, 216} \right)$ code also outperforms the ${\mathcal P}\left( {1024, 384} \right)$ code.
Note that the rate profiling in (\ref{eq3-1-4}) deviates from the optimal construction for this length.
Moreover, the ${\mathcal G}\left( {1024, 552, 216} \right)$ code provides a significant decoding latency advantage over this polar code, as will be discussed below.
This performance gain is attributed to the nested structure of GII-polar codes.
Since a GII-polar code can be seen as a polar code with the same codeword length,
this performance gain is attributed to diversity of the decoding paths of the aforesaid polar codes.
If decoding of an interleave fails, it can be projected into decoding of the nested code, enabling richer error patterns to be corrected by each interleave.
Moreover, this nested structure also enables the decoder to adapt to different realizations of channel through decoding shorter polar codes with matched code rate.
Consequently, the decoding performance of GII-polar codes can be further improved.

\subsection{Decoding Latency and Complexity}

Table \ref{table:latency} shows the decoding latencies for the ${\mathcal G}\left( {1024, 552, 216} \right)$, the ${\mathcal P}\left( {2048, 768} \right)$, and the ${\mathcal P}\left( {1024, 384} \right)$ codes.
It can be observed that the proposed GII-polar code achieves nearly half the decoding latency of a single polar code with length $2N$, while provides superior FER performance.
In particular, compared with the ${\mathcal P}\left(1024, 384\right)$ code decoded by the SCL decoding, the proposed GII-polar code attains a comparable decoding latency but achieves better FER performance. 

\begin{figure}[t]
	\centering
	\includegraphics[scale=0.32]{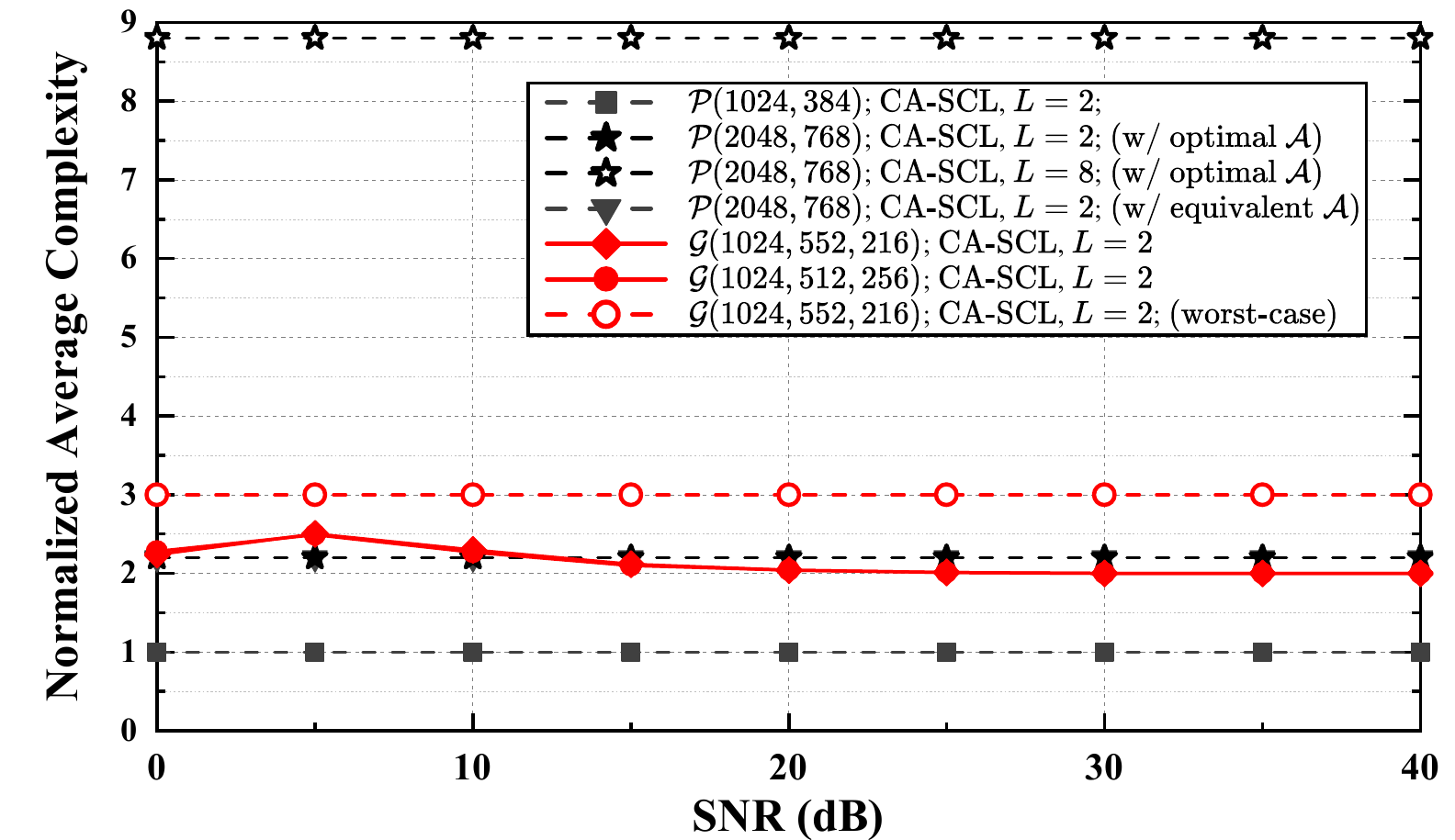}
	\caption{
		Normalized complexity in decoding the GII-polar codes.
	}
	\label{Fig.6-2}
\end{figure}

\begin{table}[t]
	\centering
	\caption{Average Latency in Decoding the GII-Polar Codes}
	\label{table:latency}
	\footnotesize
	\begin{tabular}{|l|l|c|c|c|}
		\hline
		\multirow{2}{*}{Codes} & \multirow{2}{*}{SCL Decoder} & \multicolumn{3}{c|}{SNR (dB)} \\
		\cline{3-5}
		& & 5 & 15 & 25 \\
		\hline

		\multirow{1}{*}{$\mathcal{P}(1024,384)$}

		& $L =2 $
		& 2,454 & 2,454 & 2,454 \\ \cline{2-5}
		\hline

		\multirow{1}{*}{$\mathcal{P}(1024,552)$}

		& $L =2 $
		& 2,622 & 2,622 & 2,622 \\ \cline{2-5}
		\hline

		\multirow{1}{*}{$\mathcal{P}(2048,768)$}
		& $L \in \{2, 8\} $   
		& 4,886 & 4,886 & 4,886 \\ \cline{2-5}
		\hline

		\multirow{1}{*}{$\mathcal{G}(1024,552,216)$}
		& $L =2 $
		& \textbf{3,763} & \textbf{2,888} & \textbf{2,626} \\ \cline{2-5}

		\hline


		\hline
	\end{tabular}
\end{table}

Fig. \ref{Fig.6-2} shows the normalized average complexity in decoding the ${\mathcal G}\left( {1024, 552, 216} \right)$ GII-polar code, where the normalization factor is $2N{\log _2} {N} $ and $N=1024$.
With SCL decoding, the ${\mathcal G}\left(1024, 552, 216\right)$ code exhibits a decoding complexity that is comparable to that of the ${\mathcal P}\left(2048, 768\right)$ code.
However, the proposed GII-polar code achieves superior decoding performance and significantly lower decoding latency.
Finally, it should be pointed out that, under SCL decoding, decoding complexity of the ${\mathcal G}\left(1024, 552, 216\right)$ code is higher than that of the ${\mathcal P}\left(1024, 384\right)$ code. 

\section{Conclusion}

In this paper, the GII-polar code has been proposed.
Simulation results have shown that it outperforms the conventional polar codes over the BRFC.
In particular, since decoding of all interleaves can be performed in parallel, a GII-polar code, with two consecutive polar codewords of length $N$, has yielded a reduced decoding latency over a single polar code of length $2N$.
Future work will focus on the performance analysis and the hardware design of GII-polar codes.

\section*{Acknowledgement}
This work is supported in part by the National Natural Science Foundation of China (NSFC) with project ID 62471503; and in part by the Natural Science Foundation of Guangdong Province (NSFGP) with project ID 2024A1515010213.

The authors thank {Prof.} Emanuele Viterbo (from the Department of Electrical and Computer Systems Engineering, Monash University, Melbourne, Australia) and {Prof.} Erdal Ar{\i}kan (from the Department of Electrical-Electronics Engineering, Bilkent University, Ankara, Turkey) for the fruitful discussions at ISIT 2026 that inspired this work.


\bibliographystyle{IEEEtran}      
\bibliography{IEEEabrv,ref}

\end{document}